\documentclass[aps,preprint]{revtex4}%
\usepackage{amsfonts}
\usepackage{amsmath}
\usepackage{amssymb}
\usepackage{graphicx}%
\providecommand{\U}[1]{\protect\rule{.1in}{.1in}}
\newtheorem{theorem}{Theorem}

\newtheorem{remark}[theorem]{Remark}

\newenvironment{proof}[1][Proof]{\noindent\textbf{#1.} }{\ \rule{0.5em}{0.5em}}
\begin{document}
\preprint{ }
\title[ ]{Relativistic dynamics, Green function and pseudodifferential operators}
\author{Diego Julio Cirilo-Lombardo}
\affiliation{National Institute of Plasma Physics (INFIP),}
\affiliation{Universidad de Buenos Aires, Consejo Nacional de Investigaciones Cientificas y
Tecnicas (CONICET), Facultad de Ciencias Exactas y Naturales, Department of
Physics, Ciudad Universitaria, Buenos Aires 1428, Argentina}
\affiliation{Bogoliubov Laboratory of Theoretical Physics, Joint Institute for Nuclear
Research, 141980 Dubna, Russia}
\author{}
\affiliation{}

\begin{abstract}
The central role played by pseudodifferential operators in relativistic
dynamics is very well know. In this work, operators as the Schrodinger one
(e.g: square root) are treated from the point of view of the non-local
pseudodifferential Green functions. Starting from the explicit construction of
the Green (semigroup) theoretical kernel, a theorem linking the integrability
conditions and their dependence on the spacetime dimensions is given.
Relativistic wave equations with arbitrary spin and the causality problem are
discussed with the algebraic interpretation of the radical operator and their
relation with coherent and squeezed states. Also we perform by mean of pure
theoretical procedures (based in physical concepts and symmetry) the
relativistic position operator which satisfies the conditions of integrability
: it is non-local, Lorentz invariant and does not have the same problems as
the "local"position operator proposed by Newton and Wigner. Physical examples,
as Zitterbewegung and rogue waves, are presented and deeply analysed in this
theoretical framework.

\end{abstract}
\volumeyear{year}
\volumenumber{number}
\issuenumber{number}
\eid{identifier}
\date[Date text]{date}
\received[Received text]{date}

\revised[Revised text]{date}

\accepted[Accepted text]{date}

\published[Published text]{date}

\startpage{101}
\endpage{102}
\maketitle
\tableofcontents

\section{Introduction}

Since long ago, the correct mathematical treatment of various operators that
contain a nonlinear character called the permanent attention of researchers in
both: physics and mathematics. In particular, the natural introduction of the
square root of the fundamental operators (Hamiltonian, Lagrangian) in the
field of relativistic theories has brought not only the problem of
mathematical treatment of these non-local and nonlinear operators, but its
physical interpretation also. The conceptual fact to find the physical
interpretation of the square root operator has not only addressed the
classical dynamics but in particular its action on physical states of the
system, mainly in the quantum case (spectrum). About this issue, the works of
Salpeter and others \cite{salp} , are very well known. From the purely formal
point of view, the mathematical developments has been oriented in the area of
the pseudodifferential operators with the works of \cite{semi} and in the
theory of semigroups, mainly with the research of Yosida\cite{yos1}.

Several different physical scenarios have been shown that the correct
description is made by mean these pseudodifferential operators. For example,
in\cite{seakeeping} a semi-analytical computation of the three dimensional
Green function of a pseudodifferential operator for seakeeping flow problems
is proposed where the potential flow model is assumed with harmonic dependence
on time and a linearized free surface boundary condition. Also in
\cite{carlson}, the pseudodifferential operator is introduced in Carson's
integral computation ( the original expression involves the evaluation of just
one Struve and one Bessel function) and used in power systems analysis for
evaluating the earth-return impedance of overhead conductors above homogeneous earth.

The starting point in this article is to analyze the classical radical
operator from the viewpoint of Green functions connecting these results with
previous investigations in which we have related these type of operators with
the well known problem of localization and proper time. In second place, we
will show the purely algebraic representation giving an interpretation of such
operators. This algebraic interpretation is very important because brings the
possibility to make a link with pseudodifferential operators and semigroup
(Fourier-Integral) representations:
\[%
\begin{array}
[c]{ccccc}
&  &
\begin{array}
[c]{c}%
algebraic\\
interpretation
\end{array}
&  & \\
& \nearrow &  & \nwarrow & \\
pseudodifferential\text{ }operators &  & \longleftrightarrow &  &
\begin{array}
[c]{c}%
semigroup\\
(Fourier-Integral\text{ }representations)
\end{array}
\end{array}
\]
and consequently with the relativistic wave equations of \ arbitrary spin
states (in particular of parastatistical ones). Several examples, from our
earlier references on the subject (e.g. \cite{epjc} \cite{meta}), and new ones
(as the important problem of rogue waves) are analyzed and discussed. Finally,
we discuss other important questions that will be treated in detail somewhere
\cite{diegoN} as the role of the spin, the dispersion in \ time of the
physical states and the Levy processes.

\footnote[1]{It is useful to know that we can adimensionalize our problem
introducing $\tau\equiv\frac{ct}{\lambda_{c}},$ $\eta\equiv\frac{\hbar x}{mc}$
$\left(  \lambda_{c}=\frac{h}{mc}\right)  $ then, the Schrodinger equation
becomes to $\partial_{\tau}\psi\left(  \tau,\eta\right)  =\sqrt{1-\partial
_{\eta}^{2}}\psi\left(  \tau,\eta\right)  .$In other case, we take through
this paper: $\hbar=1=c$}

\section{Green functions and pseudodifferential operators}

The Schrodinger operator is, in general, the main object of study both: in
theoretical physics and mathematics and from the classical to the quantum
point of view. However, from the operator viewpoint, the relativistic form
(e.g. square root ) carry conceptual and technical troubles. As we have been
discussed before\cite{diego1}\cite{diego2}\cite{diego3}, the conceptual
trouble coming from three sources: the meaning of the Lagrangian as
\textit{measure}, the \textit{localization} (one particle/ensemble
interpretation) and the relation with \textit{spin degrees of freedom};
meanwhile the technical one is given by the square root form of the
Lagrangian/Hamiltonian when is treated as \textit{operator}. Consequently, in
this paper our starting point will be the following simple mathematical object
(principal kernel):%
\begin{equation}
\mathbb{L}\equiv\sqrt{\zeta^{2}+m^{2}} \tag{1}%
\end{equation}
that is characteristic inside the fundamental Schrodinger type equation to
solve ($g_{\mu\nu}=\left(  +---\right)  ;\mu,\nu=0,1,2,3)),$ namely%
\begin{equation}
i\partial_{t}\psi\left(  t,x\right)  =\frac{1}{\left(  2\pi\right)  ^{3}}%
\int\int\sqrt{\zeta^{2}+m^{2}}e^{i\zeta\cdot\left(  x-y\right)  }\psi\left(
t,y\right)  d^{3}yd^{3}\zeta\equiv\mathcal{H}\psi\left(  t,x\right)  \tag{2}%
\end{equation}
and having account that the square root operator above is well defined in such
pseudodifferential representation due its strongly elliptness (for instance
$\zeta^{2}\equiv\gamma_{ij}\zeta^{i}\zeta^{j}$ with $i,j=1,2,3$), the solution
for the source which is a delta function becomes:%
\begin{equation}
i\partial_{t}G\left(  t,x\right)  -\frac{1}{\left(  2\pi\right)  ^{3}}\int
\int\sqrt{\zeta^{2}+m^{2}}e^{i\zeta\cdot\left(  x-y\right)  }G\left(
t,y\right)  d^{3}yd^{3}\zeta\equiv\mathcal{-\delta}^{3}\left(  x\right)
\mathcal{\delta}\left(  t\right)  \tag{3}%
\end{equation}
We perform the Fourier transform obtaining%
\begin{equation}
\frac{1}{\left(  2\pi\right)  ^{2}}\int(\omega G\left(  \omega,p\right)
-\sqrt{p^{2}+m^{2}}G\left(  \omega,p\right)  )e^{-i\left(  \omega t-p\cdot
x\right)  }d^{3}pd\omega=-\frac{1}{\left(  2\pi\right)  ^{4}}\int e^{-i\left(
\omega t-p\cdot x\right)  }d^{3}pd\omega\tag{4}%
\end{equation}
consequently we have in the momentum space
\begin{equation}
G\left(  \omega,p\right)  =-\frac{1}{\left(  2\pi\right)  ^{2}}\frac{1}%
{\omega-\sqrt{p^{2}+m^{2}}} \tag{5}%
\end{equation}
Anti-transforming back to $x-$space we obtain%
\begin{equation}
G\left(  t,x\right)  =-\frac{1}{\left(  2\pi\right)  ^{4}}\int\frac
{e^{-i\left(  \omega t-p\cdot x\right)  }}{\omega-\sqrt{p^{2}+m^{2}}}%
d^{3}pd\omega\tag{6}%
\end{equation}
Choosing a path over the complex plane, the integration over $\omega$ can be
performed. Then we replace the denominator by $\omega-\left(  1\mp
i\varepsilon\right)  \sqrt{p^{2}+m^{2}}$ with $\varepsilon>0$ evidently.
Integration along a in the upper half plane gives $G^{+}\left(  t,x\right)  $
and the other one $G^{-}\left(  t,x\right)  :$%
\begin{align}
G^{+}\left(  t,x\right)   &  =\frac{i}{\left(  2\pi\right)  ^{3}}%
\underset{\varepsilon\rightarrow0}{\lim}\theta\left(  t\right)  \int
e^{-i\left(  \left(  1-i\varepsilon\right)  \sqrt{p^{2}+m^{2}}t-p\cdot
x\right)  }d^{3}p\tag{7}\\
G^{-}\left(  t,x\right)   &  =-\frac{i}{\left(  2\pi\right)  ^{3}}%
\underset{\varepsilon\rightarrow0}{\lim}\theta\left(  -t\right)  \int
e^{-i\left(  \left(  1+i\varepsilon\right)  \sqrt{p^{2}+m^{2}}t-p\cdot
x\right)  }d^{3}p \tag{8}%
\end{align}
where $\theta\left(  t\right)  $ is the Theta (step) function. Notice that
$G_{D}=G^{+}\left(  t,x\right)  -G^{-}\left(  t,x\right)  =\frac{i}{\left(
2\pi\right)  ^{3}}\int e^{-i\left(  \sqrt{p^{2}+m^{2}}t-p\cdot x\right)
}d^{3}p.$

Integration with respect to $d^{3}p$ is performed as follows: we write in the
spherical prescription only for simplicity, $p\cdot x=\left\vert p\right\vert
\left\vert x\right\vert \cos\theta$ and integrating with respect to $\varphi$
and $\theta$%
\begin{align}
G^{\pm}\left(  t,x\right)   &  =\pm\frac{i}{\left(  2\pi\right)  ^{3}%
}\underset{\varepsilon\rightarrow0}{\lim}\theta\left(  \pm t\right)  \int
e^{-i\left(  \left(  1\mp i\varepsilon\right)  \sqrt{p^{2}+m^{2}}t-\left\vert
p\right\vert \left\vert x\right\vert \cos\theta\right)  }p^{2}dp\sin\theta
d\theta d\varphi\tag{9}\\
&  =\pm\frac{2i}{\left(  2\pi\right)  ^{2}}\underset{\varepsilon\rightarrow
0}{\lim}\theta\left(  \pm t\right)  \int e^{-it\left(  1\mp i\varepsilon
\right)  \sqrt{p^{2}+m^{2}}}pdp\sin\left(  \left\vert p\right\vert \left\vert
x\right\vert \right) \nonumber\\
&  =\pm\frac{2i}{\left(  2\pi\right)  ^{2}}\underset{\varepsilon\rightarrow
0}{\lim}\theta\left(  \pm t\right)  it\left(  1\mp i\varepsilon\right)
m^{2}\frac{K_{2}\left(  m\sqrt{x^{2}-t^{2}\left(  1\mp i\varepsilon\right)
^{2}}\right)  }{x^{2}-t^{2}\left(  1\mp i\varepsilon\right)  ^{2}} \tag{11}%
\end{align}
where we have been used\cite{GR} (formulas 3.915)

\section{Relation with the operational approach: example in one dimension}

From the fundamental equation to solve ($g_{\mu\nu}=\left(  +---\right)
;\mu,\nu=0,1,2,3))$%
\begin{equation}
i\partial_{t}\psi\left(  t,x\right)  =\frac{1}{\left(  2\pi\right)  ^{3}}%
\int\int\sqrt{\zeta^{2}+m^{2}}e^{i\zeta\cdot\left(  x-y\right)  }\psi\left(
t,y\right)  dyd\zeta\equiv\mathcal{H}\psi\left(  t,x\right)  \tag{12}%
\end{equation}
and having account that the square root operator above is well defined in such
pseudodifferential representation due its strongly elliptness (for instance
$\zeta^{2}\equiv\gamma_{ij}\zeta^{i}\zeta^{j}$ with i,j=1,2,3), the solution
for the source becomes%
\begin{equation}
i\partial_{t}G\left(  t,x\right)  -\frac{1}{\left(  2\pi\right)  ^{3}}\int
\int\sqrt{\zeta^{2}+m^{2}}e^{i\zeta\cdot\left(  x-y\right)  }G\left(
t,y\right)  dyd\zeta\equiv\mathcal{-\delta}\left(  x\right)  \mathcal{\delta
}\left(  t\right)  \tag{13}%
\end{equation}
We perform the Fourier transform obtaining%
\begin{equation}
\int(\omega G\left(  \omega,p\right)  -\sqrt{p^{2}+m^{2}}G\left(
\omega,p\right)  )e^{-i\left(  \omega t-p\cdot x\right)  }dpd\omega=-\frac
{1}{\left(  2\pi\right)  ^{2}}\int e^{-i\left(  \omega t-p\cdot x\right)
}dpd\omega\tag{14}%
\end{equation}
consequently we have in the momentum space
\begin{equation}
G\left(  \omega,p\right)  =-\frac{1}{\omega-\sqrt{p^{2}+m^{2}}} \tag{15}%
\end{equation}
Anti-transforming back to $x-$space we have%
\begin{equation}
G\left(  t,x\right)  =-\frac{1}{\left(  2\pi\right)  ^{2}}\int\frac
{e^{-i\left(  \omega t-p\cdot x\right)  }}{\omega-\sqrt{p^{2}+m^{2}}}%
dpd\omega\tag{16}%
\end{equation}
Again, choosing a path over the complex plane, the integration over $\omega$
can be performed. Then, we replace the denominator by $\omega-\left(  1\mp
i\varepsilon\right)  \sqrt{p^{2}+m^{2}}$ with $\varepsilon>0$ as is evident.
Integration along a in the upper half plane gives $G^{+}\left(  t,x\right)  $
and the other one $G^{-}\left(  t,x\right)  :$%
\begin{align}
G^{+}\left(  t,x\right)   &  =\frac{i}{2\pi}\underset{\varepsilon\rightarrow
0}{\lim}\theta\left(  t\right)  \int e^{-i\left(  \left(  1-i\varepsilon
\right)  \sqrt{p^{2}+m^{2}}t-p\cdot x\right)  }dp\tag{17}\\
G^{-}\left(  t,x\right)   &  =-\frac{i}{2\pi}\underset{\varepsilon
\rightarrow0}{\lim}\theta\left(  -t\right)  \int e^{-i\left(  \left(
1+i\varepsilon\right)  \sqrt{p^{2}+m^{2}}t-p\cdot x\right)  }dp \tag{18}%
\end{align}
where $\theta\left(  t\right)  $ is the Theta (step) function. Notice that:
$G^{+}\left(  t,x\right)  -G^{-}\left(  t,x\right)  =\frac{i}{2\pi}\int
e^{-i\left(  \sqrt{p^{2}+m^{2}}t-p\cdot x\right)  }dp.$

Integration with respect to $dp$ is performed directly obtaining%
\begin{equation}
G^{\pm}\left(  t,x\right)  =\pm\frac{i}{2\pi}\underset{\varepsilon
\rightarrow0}{\lim}\theta\left(  \pm t\right)  it\left(  1\mp i\varepsilon
\right)  m\frac{K_{1}\left(  m\sqrt{x^{2}-t^{2}\left(  1\mp i\varepsilon
\right)  ^{2}}\right)  }{\sqrt{x^{2}-t^{2}\left(  1\mp i\varepsilon\right)
^{2}}} \tag{19}%
\end{equation}
where we have been used\cite{GR} ( formulas 3.914).

Now we will make the proof about the direct relation between the operational
approach and the Green function one.

We know that:%

\begin{align}
LG\left(  t,x\right)   &  \equiv i\partial_{t}G\left(  t,x\right)  -\frac
{1}{\left(  2\pi\right)  ^{3}}\int\int\sqrt{\zeta^{2}+m^{2}}e^{i\zeta
\cdot\left(  x-y\right)  }G\left(  t,y\right)  dyd\zeta\tag{20}\\
&  \equiv\mathcal{-\delta}\left(  x\right)  \mathcal{\delta}\left(  t\right)
\tag{21}%
\end{align}
then, the solution given by the difference of the fundamental Green functions
is solution of the free Salpeter equation with constant initial condition:
\begin{equation}
LG_{D}=L\left(  G^{+}\left(  t,x\right)  -G^{-}\left(  t,x\right)  \right)  =0
\tag{22}%
\end{equation}
multiplication with the general initial condition as ($a$ is some constant):%
\begin{equation}
\Psi_{0}\left(  p\right)  \equiv a\int e^{-i\left(  \sqrt{p^{\prime2}+m^{2}%
}\beta-p^{\prime}\alpha\right)  }\delta\left(  p-p^{\prime}\right)
dp^{\prime} \tag{23}%
\end{equation}
is also solution. It is useful to note that $\Psi_{0}\left(  p\right)  $ acts
as convolution due its reproducing properties (notice the obvious fact that
have the same fashion that the kernel for the square root transformation as in
the Schwartzian case) making a \ general shift in time and space:%

\begin{align}
\left(  G^{+}\left(  t,x\right)  -G^{-}\left(  t,x\right)  \right)  \Psi
_{0}\left(  p\right)   &  =\frac{i}{2\pi}\int e^{-i\left(  \sqrt{p^{2}+m^{2}%
}\left(  t+\beta\right)  -p\left(  x+\alpha\right)  \right)  }dp=0\tag{24}\\
&  =\frac{1}{2\pi}\left(  t+\beta\right)  m\frac{K_{1}\left(  m\sqrt{\left(
x+\alpha\right)  ^{2}-\left(  t+\beta\right)  ^{2}}\right)  }{\sqrt{\left(
x+\alpha\right)  ^{2}-\left(  t+\beta\right)  ^{2}}} \tag{25}%
\end{align}

\bigskip Then, the suitable initial condition for the vacuum acts also as
potential\ in the Green's method (for any $G^{\pm}\left(  t,x\right)  $):%
\begin{gather}
\int(\omega G^{\pm}\left(  \omega,p\right)  -\sqrt{p^{2}+m^{2}}G^{\pm}\left(
\omega,p\right)  )e^{-i\left(  \omega t-p\cdot x\right)  }e^{-i\left(
\sqrt{p^{2}+m^{2}}\beta-p\alpha\right)  }dpd\omega=\tag{26}\\
=-\frac{1}{\left(  2\pi\right)  ^{2}}\int e^{-i\left(  \omega t-p\cdot
x\right)  }dpd\omega e^{-i\left(  \sqrt{p^{2}+m^{2}}\beta-p\alpha\right)
}\nonumber
\end{gather}

\section{The dimension-dependence of the Green function}

\begin{theorem}
The Green function of the square root pseudodifferential operator in the case
of even dimensional spacetime is exactly integrable being expressed in a
closed form as derivatives of the Mac Donald's function $K_{2}$ (Modified
Bessel Function of the Second Kind).
\end{theorem}

\begin{proof}
Following the same procedure as before, we introduce hyperspherical
coordinates as usual:%
\begin{align}
x_{1}  &  =r\cos\phi_{1}\tag{27}\\
x_{2}  &  =r\sin\phi_{1}\cos\phi_{2}\nonumber\\
x_{3}  &  =r\sin\phi_{1}\sin\phi_{2}\cos\phi_{3}\nonumber\\
&  \cdot\nonumber\\
&  \cdot\nonumber\\
&  \cdot\nonumber\\
x_{N}  &  =r\sin\phi_{1}\sin\phi_{2}\sin\phi_{3}\cdot\cdot\cdot\cdot\sin
\phi_{N-1}\cos\phi_{N}\nonumber
\end{align}
consequently, our problem will be the following integral%
\begin{gather}
\int e^{-i\left(  \left(  1\mp i\varepsilon\right)  \sqrt{p^{2}+m^{2}%
}t-\left\vert p\right\vert \left\vert x\right\vert \cos\theta_{1}\right)
}p^{N-1}dp\sin\theta_{1}d\theta_{1}\sin\theta_{2}d\theta_{2}\cdot\cdot
\cdot\sin\theta_{N-1}d\theta_{N-1}d\varphi=\tag{28}\\
=F\underset{\equiv I}{\underbrace{\int e^{-i\left(  \left(  1\mp
i\varepsilon\right)  \sqrt{p^{2}+m^{2}}t-\left\vert p\right\vert \left\vert
x\right\vert \cos\theta_{1}\right)  }p^{N-1}dp\sin\theta_{1}d\theta_{1},}%
}\nonumber
\end{gather}
where is a numerical factor coming from the: $\sin\theta_{2}d\theta_{2}%
\cdot\cdot\cdot\sin\theta_{N-1}d\theta_{N-1}d\varphi$ integration.%
\begin{equation}
F=\int\sin\theta_{2}d\theta_{2}\cdot\cdot\cdot\sin\theta_{N-1}d\theta
_{N-1}d\varphi=2\pi2^{N-2} \tag{29}%
\end{equation}
Let us to compute the $I$ integral, integrating in $\theta_{1}$first:%
\begin{align}
I  &  =\int e^{-i\left(  \left(  1\mp i\varepsilon\right)  \sqrt{p^{2}+m^{2}%
}t-\left\vert p\right\vert \left\vert x\right\vert \cos\theta_{1}\right)
}p^{N-1}dp\left(  -d\cos\theta_{1}\right) \tag{30}\\
&  =\int e^{-i\left(  1\mp i\varepsilon\right)  \sqrt{p^{2}+m^{2}}t}\text{
}p^{N-1}dp\left(  \left.  -\frac{e^{-i\left\vert p\right\vert \left\vert
x\right\vert \cos\theta_{1}}}{i\left\vert p\right\vert \left\vert x\right\vert
\cos\theta_{1}}\right\vert _{\theta_{i}=0}^{\theta_{f=}\pi}\right) \nonumber\\
&  =\int e^{-i\left(  1\mp i\varepsilon\right)  \sqrt{p^{2}+m^{2}}t}\text{
}\frac{p^{N-1}}{i\left\vert p\right\vert \left\vert x\right\vert }dp\left(
-e^{-i\left\vert p\right\vert \left\vert x\right\vert }+e^{i\left\vert
p\right\vert \left\vert x\right\vert }\right) \nonumber\\
&  =\int e^{-i\left(  1\mp i\varepsilon\right)  \sqrt{p^{2}+m^{2}}t}\text{
}\frac{p^{N-2}}{\left\vert x\right\vert }dp\text{ }2\sin\left(  \left\vert
p\right\vert \left\vert x\right\vert \right) \nonumber
\end{align}
We have been solved the case for $d=2(N=1)$ and $d=4(N=3)$. We can demonstrate
the integrability of the higher dimensional family of cases where the
spacetime dimension is even $d=1+N$. The starting point is to use the
formula:
\begin{align}
p^{N-2}\sin\left(  \left\vert p\right\vert \left\vert x\right\vert \right)
&  =\left(  \frac{id}{dx}\right)  ^{N-3}\left(  p\sin\left(  \left\vert
p\right\vert \left\vert x\right\vert \right)  \right) \tag{31}\\
e.g  &  :p^{5-2}\sin\left(  \left\vert p\right\vert \left\vert x\right\vert
\right)  =-\frac{d^{5-3}}{dx}\left(  p\sin\left(  \left\vert p\right\vert
\left\vert x\right\vert \right)  \right) \nonumber\\
p^{3}\sin\left(  \left\vert p\right\vert \left\vert x\right\vert \right)   &
=-\frac{d^{2}}{dx}\left(  p\sin\left(  \left\vert p\right\vert \left\vert
x\right\vert \right)  \right)  =-p^{2}\frac{d}{dx}\left(  \cos\left(
\left\vert p\right\vert \left\vert x\right\vert \right)  \right)  =+p^{3}%
\sin\left(  \left\vert p\right\vert \left\vert x\right\vert \right) \nonumber
\end{align}%
\begin{align*}
e.g  &  :p^{7-2}\sin\left(  \left\vert p\right\vert \left\vert x\right\vert
\right)  =+\frac{d^{7-3}}{dx^{7-3}}\left(  p\sin\left(  \left\vert
p\right\vert \left\vert x\right\vert \right)  \right) \\
p^{5}\sin\left(  \left\vert p\right\vert \left\vert x\right\vert \right)   &
=p^{2}\frac{d^{3}}{dx^{3}}\left(  \cos\left(  \left\vert p\right\vert
\left\vert x\right\vert \right)  \right)  =-p^{3}\frac{d^{2}}{dx^{2}}\left(
\sin\left(  \left\vert p\right\vert \left\vert x\right\vert \right)  \right)
=-p^{4}\frac{d}{dx}\cos\left(  \left\vert p\right\vert \left\vert x\right\vert
\right)  =\\
&  =+p^{5}\sin\left(  \left\vert p\right\vert \left\vert x\right\vert \right)
\end{align*}
then , the expression (31) becomes to:%
\begin{align}
I  &  =\int e^{-i\left(  1\mp i\varepsilon\right)  \sqrt{p^{2}+m^{2}}t}\text{
}\frac{p^{N-2}}{\left\vert x\right\vert }dp\text{ }2\sin\left(  \left\vert
p\right\vert \left\vert x\right\vert \right)  =\tag{32}\\
&  =\frac{2}{\left\vert x\right\vert }\left(  \frac{id}{dx}\right)
^{N-3}\left[  \int e^{-i\left(  1\mp i\varepsilon\right)  \sqrt{p^{2}+m^{2}}%
t}\text{ }pdp\text{ }\sin\left(  \left\vert p\right\vert \left\vert
x\right\vert \right)  \right]  =\nonumber\\
&  =2\left(  \frac{id}{dx}\right)  ^{N-3}\left[  \pm\frac{2i}{\left(
2\pi\right)  ^{2}}\underset{\varepsilon\rightarrow0}{\lim}\theta\left(  \pm
t\right)  it\left(  1\mp i\varepsilon\right)  m^{2}\frac{K_{2}\left(
m\sqrt{x^{2}-t^{2}\left(  1\mp i\varepsilon\right)  ^{2}}\right)  }%
{x^{2}-t^{2}\left(  1\mp i\varepsilon\right)  ^{2}}\right] \nonumber
\end{align}
consequently for the case of $(2M+1,1)$ dimensions (e.g. even dimensions) we
have%
\begin{equation}
G^{\pm}\left(  t,x\right)  =\pi2^{N}\left(  \frac{id}{dx}\right)
^{N-3}\left[  \pm\frac{2i}{\left(  2\pi\right)  ^{2}}\underset{\varepsilon
\rightarrow0}{\lim}\theta\left(  \pm t\right)  it\left(  1\mp i\varepsilon
\right)  m^{2}\frac{K_{2}\left(  m\sqrt{x^{2}-t^{2}\left(  1\mp i\varepsilon
\right)  ^{2}}\right)  }{x^{2}-t^{2}\left(  1\mp i\varepsilon\right)  ^{2}%
}\right]  \tag{33}%
\end{equation}

\begin{remark}
The integrability in even dimensions is due to an underlying symplectic
structure ("phase space map", only valid in even number of spacetime dimensions).
\end{remark}
\end{proof}

\section{Relativistic position operator from velocity operator}

\subsection{Velocity operator: definition and action}

Canonically the reasonable assumption as our starting point are the following relations:%

\begin{equation}
H\left(  x,p\right)  \equiv\sqrt{p^{2}+m^{2}} \tag{34}%
\end{equation}
and
\begin{align}
\frac{dp}{dt}  &  =-\frac{\partial H\left(  x,p\right)  }{\partial x}%
\tag{35}\\
\frac{dx}{dt}  &  =\frac{\partial H\left(  x,p\right)  }{\partial p} \tag{36}%
\end{align}
Then, the velocity operator is defined from:%
\begin{align}
\widehat{u}^{\alpha}  &  =\frac{d\widehat{x}^{\alpha}}{dt}=-i\left[  H\left(
x,p\right)  ,\widehat{x}^{\alpha}\right] \tag{37}\\
&  =\frac{\partial H\left(  x,p\right)  }{\partial p_{\alpha}}=\frac
{\delta^{\alpha\beta}p_{\beta}}{\sqrt{p^{2}+m^{2}}} \tag{38}%
\end{align}
with the action as%
\begin{equation}
\left(  \widehat{u}^{\alpha}\psi\right)  \left(  t,x\right)  =\frac{1}{\left(
2\pi\right)  ^{3}}\int\int\widehat{u}^{\alpha}\left(  \zeta\right)
e^{i\zeta\cdot\left(  x-y\right)  }\psi\left(  t,y\right)  d^{3}yd^{3}%
\zeta\tag{39}%
\end{equation}%
\begin{align}
\left(  \widehat{u}^{\alpha}\psi\right)  \left(  t,x\right)   &  =\frac
{1}{\left(  2\pi\right)  ^{3}}\int\int\frac{\zeta^{\alpha}}{\sqrt{\zeta
^{2}+m^{2}}}e^{i\zeta\cdot\left(  x-y\right)  }\psi\left(  t,y\right)
d^{3}yd^{3}\zeta\tag{40}\\
&  =\frac{1}{\left(  2\pi\right)  ^{3}}\int\int\frac{\zeta^{\alpha}}%
{\sqrt{\zeta^{2}+m^{2}}}e^{i\zeta\cdot\left(  x-y\right)  }\psi\left(
t,y\right)  d^{3}yd^{2}\zeta\nonumber
\end{align}
In 1+1 dimensions we have
\begin{equation}
\left(  \widehat{u}^{\alpha}\psi\right)  \left(  t,x\right)  =\frac{1}{2\pi
}\int\int\frac{\zeta^{\alpha}}{\sqrt{\zeta^{2}+m^{2}}}e^{i\zeta\cdot\left(
x-y\right)  }\psi\left(  t,y\right)  dyd\zeta\tag{41}%
\end{equation}
and knowing that :%
\begin{equation}
\left.  \psi\left(  t,y\right)  \right\vert _{1+1}=\frac{1}{2\pi}tm\frac
{K_{1}\left(  m\sqrt{y^{2}-t^{2}}\right)  }{\sqrt{y^{2}-t^{2}}} \tag{42}%
\end{equation}
we insert it in (41), then, the equation to solve becomes to:%
\begin{equation}
\left(  \widehat{u}^{\alpha}\psi\right)  \left(  t,x\right)  =\frac{1}{2\pi
}\int\int\frac{\zeta^{\alpha}}{\sqrt{\zeta^{2}+m^{2}}}e^{i\zeta\cdot\left(
x-y\right)  }\frac{1}{2\pi}tm\frac{K_{1}\left(  m\sqrt{y^{2}-t^{2}}\right)
}{\sqrt{y^{2}-t^{2}}}dyd\zeta\tag{43}%
\end{equation}
Using formula GR 3.365 (2) in the $m\neq0$ case (the non-massive case will be
analyzed separately), the expression to integrate in $y$ is:%
\begin{equation}
\left(  \widehat{u}^{\alpha}\psi\right)  \left(  t,x\right)  =\frac{1}{\left(
2\pi\right)  ^{2}}\int\int im^{2}K_{1}\left(  m(x-y)\right)  t\frac
{K_{1}\left(  m\sqrt{y^{2}-t^{2}}\right)  }{\sqrt{y^{2}-t^{2}}}dy \tag{44}%
\end{equation}
consequently, the action of the velocity operator on the state solution given
by the "square root" Hamiltonian is (using GR\ 7$^{8})$%
\begin{align}
\left(  \widehat{u}^{\alpha}\psi\right)  \left(  t,x\right)   &  =\frac
{1}{\left(  2\pi\right)  ^{2}}\frac{m^{2}t^{2}}{\sqrt{2}}K_{1}\left(
m\sqrt{2\left(  x^{2}-t^{2}\right)  }\right) \tag{45}\\
&  \left(  K_{-\nu}=K_{\nu}\right) \nonumber
\end{align}
The extension to more dimensions is straighforward.

\subsection{Determination of the relativistic position operator}

\begin{theorem}
The canonical position operator in $p$-representation, namely $\partial_{p}$ ,
acting on the convoluted state (23)(equivalent to the initial condition in the
operatorial approach) determines(in the case of null eigenvalue) univoquely
and simultaneously the action of the velocity operator in the $p$
representation plus the ground state of the physical system under consideration.
\end{theorem}

\begin{proof}
From the initial condition (23)\ (the kernel can be straighforwardly extended
to any dimension) we have%
\begin{equation}
\Psi_{0}\left(  p\right)  \equiv ae^{-i\left(  \sqrt{p^{\prime2}+m^{2}}%
\beta-p^{\prime}\alpha\right)  } \tag{46}%
\end{equation}
Operating with $\widehat{x}\rightarrow\widehat{\partial}_{p}$ (e.g.: momentum
representation) we can see that:%
\begin{equation}
\partial_{p}\Psi_{0}\left(  p\right)  =-i\left[  \frac{p\beta}{\sqrt
{p^{2}+m^{2}}}-\alpha\right]  \Psi_{0}\left(  p\right)  \tag{47}%
\end{equation}
Then:
\begin{equation}
\partial_{p}\Psi_{0}\left(  p\right)  =0\rightarrow\frac{p\beta}{\sqrt
{p^{2}+m^{2}}}\Psi_{0}\left(  p\right)  =\alpha\Psi_{0}\left(  p\right)
\tag{48}%
\end{equation}
we obtain the action of the relativistic velocity operator (39-40) analyzed previously.
\end{proof}

Consequently we can arrive to the following:

\begin{remark}
Conversely, the respective eigenvalues of the relativistic velocity operator
are obtained from the null eigenvalue of the (standard) canonical position
operator in the momentum representation.
\end{remark}

\section{Discussion}

The meaning of the convoluted initial state $\Psi_{0}\left(  p\right)  $ can
be interpreted as follows: the conjugate variable to $x^{0}\left(
\text{identified with the \textit{physical time}}\right)  $%
\begin{equation}
p^{0}=-i\partial^{0} \tag{49}%
\end{equation}
is no longer well defined because clearly it must be expressed in terms of the
remaining variables. Then, the following identification is immediately
performed
\begin{align}
\Psi_{0}\left(  p\right)   &  \equiv ae^{-i\left(  \sqrt{p^{\prime2}+m^{2}%
}\beta-p^{\prime}\alpha\right)  }\tag{50}\\
&  =ae^{-i\left(  P^{0}\beta-p^{i}\alpha_{i}\right)  } \tag{51}%
\end{align}
where $P^{0}=\sqrt{-\partial^{i}\partial^{i}+m^{2}}$ and $p^{i}=-i\partial
^{i}$ at $x^{0}=0,$ with the consequence that the remaining generators of the
Lorentz group can be determined:%
\begin{align}
M^{ij}  &  =i\left(  x^{j}\partial^{i}-x^{i}\partial^{j}\right) \tag{52}\\
M^{0i}  &  =-\frac{1}{2}\left\{  x^{i},P^{0}\right\}  \tag{53}%
\end{align}
including the fact that $x^{i}$ and $P^{0}$ not commute. The generators on the
initial surface $x^{0}=0$(with conjugate variable $p^{0}$) split into
\textit{kinematical generators} $M^{ij}(rotations)$, $p^{i}(momenta)$ and the
\textit{dynamical} ones $M^{0i}(boost)$and $P^{0}(Hamiltonian)$ that displaces
the system away from the initial surface. From the algebraic point of view we
can check from the definition of the velocity operator:%
\begin{align}
\widehat{u}^{\alpha}  &  =\frac{d\widehat{x}^{\alpha}}{dt}=-i\left[  H\left(
x,p\right)  ,\widehat{x}^{\alpha}\right]  =\frac{\partial H\left(  x,p\right)
}{\partial p_{\alpha}}=\frac{\delta^{\alpha\beta}p_{\beta}}{\sqrt{p^{2}+m^{2}%
}}\tag{54}\\
&  =-\frac{1}{2}iM^{0i} \tag{55}%
\end{align}
that is directly related with the boost generator (notice the interplay
between the Poisson and quantum structure). As is easy to interpret, if we
have into account the spinorial (Clifford) structure of the double covering of
the Lorentz group, namely the $SL\left(  2\mathbb{C}\right)  $, the spin
degrees of freedom can be introduced (besides the orbital part): \
\begin{align}
M^{ij}  &  =i\left(  x^{j}\partial^{i}-x^{i}\partial^{j}\right)
+\varepsilon^{ijk}S^{k}\tag{56}\\
M^{0i}  &  =-\frac{1}{2}\left\{  x^{i},P^{0}\right\}  -\frac{i\varepsilon
^{ijk}\partial^{j}S^{k}}{P^{0}\pm m} \tag{57}%
\end{align}
however, if we also add the space time translations: $x^{\mu}\rightarrow$
$x^{\mu}+a^{\mu},$the Poincare group (as semidirect product of the Lorentz
group plus the space-time translations) acts on Hilbert states labeled by
vectors of the form%
\begin{equation}
\left\vert p^{i};m,s,s_{3}\right\rangle \tag{58}%
\end{equation}
which are interpreted as physical states (particles) with mass$m$, spin $s$,
3-momentum $p^{i}$ and magnetic quantum number $s_{3}.$ The positive or
negative energy depends on the sign of $\beta$ into the exponential of
$\Psi_{0}\left(  p\right)  .$ By the way, notice that this fact is connected
with the Lagrange multiplier prescription of P.A.M. Dirac\cite{dir1}where%
\begin{equation}
P^{0}=p^{0}+\lambda\left(  p^{\mu}p_{\mu}+m^{2}\right)  \tag{59}%
\end{equation}
then, eliminating $p^{0}$ through $\lambda$ (e.g. taking $P^{0}$ as
Hamiltonian) we have: $P^{0}=p^{0}=\pm\sqrt{p^{i}p^{i}+m^{2}}$ corresponding
to positive and negative energy solutions.

The familiar interpretation of the eigenvalues of an observable as the only
possible values that can result from measurements of the observable on any
state of the system is no longer tenable, because the expectation value of an
observable in a particular state and the average of the eigenvalues of the
observable weighted by the absolute square of the amplitude of the
corresponding eigenstate in the state under question are not equal. A simple
but telling example is the case of the operator of the Sakata-Taketani: the
velocity operator in the conventional language, has only zero eigenvalues, yet
it is an observable in the sense of pseudo-hermiticity and has a non-vanishing
expectation value in an arbitrary state. As was pointed out before\cite{Matt},
this kind of difficulty connected with the appearance of the indefinite metric
renders the choice between different possible operators for an observable
quantity much more difficult than in the spin 1/2 case.

\section{Mp$\left(  n\right)  $ and the algebraic interpretation of the square
root operator}

Geometrically, in our early work \cite{diego1}, we take as the starting point
the action functional that will describe the world-line of the superparticle
(measure on a superspace) as follows:
\begin{equation}
S=\int_{\tau_{1}}^{\tau_{2}}d\tau L\left(  x,\theta,\overline{\theta}\right)
=-m\int_{\tau_{1}}^{\tau_{2}}d\tau\sqrt{\overset{\circ}{\omega_{\mu}}%
\overset{\circ}{\omega^{\mu}}+{\mathbf{a}}\overset{.}{\theta}^{\alpha}%
\overset{.}{\theta}_{\alpha}-{\mathbf{a}}^{\ast}\overset{.}{\overline{\theta}%
}^{\overset{.}{\alpha}}\overset{.}{\overline{\theta}}_{\overset{.}{\alpha}}}
\tag{60}%
\end{equation}
where $\overset{\circ}{\omega_{\mu}}=\overset{.}{x}_{\mu}-i(\overset{.}%
{\theta}\ \sigma_{\mu}\overline{\theta}-\theta\ \sigma_{\mu}\overset
{.}{\overline{\theta}})$, and the dot indicates derivative with respect to the
parameter $\tau$, as usual. The above Lagrangian (we will not give the details
here) was constructed considering the line element (e.g.: measure, positive
square root of the interval) of the non-degenerated supermetric introduced in
\cite{diego1} $ds^{2}=\omega^{\mu}\omega_{\mu}+{\mathbf{a}}\omega^{\alpha
}\omega_{\alpha}-{\mathbf{a}}^{\ast}\omega^{\dot{\alpha}}\omega_{\dot{\alpha}%
},$ where the bosonic term and the Majorana bispinor compose a superspace
$(1,3|1)$, with coordinates $(t,x^{i},\theta^{\alpha},\bar{\theta}%
^{\dot{\alpha}})$, and where Cartan forms of supersymmetry group are described
by: $\omega_{\mu}=dx_{\mu}-i(d\theta\sigma_{\mu}\bar{\theta}-\theta\sigma
_{\mu}d\bar{\theta}),\qquad\omega^{\alpha}=d\theta^{\alpha},\qquad\omega
^{\dot{\alpha}}=d\theta^{\dot{\alpha}}$ (obeying evident supertranslational
invariance). \footnote[1]{As we have extended our manifold to include
fermionic coordinates, it is natural to extend also the concept of trajectory
of point particle to the superspace.} To do this, we take the coordinates
$x\left(  \tau\right)  $, $\theta^{\alpha}\left(  \tau\right)  $ and
$\overline{\theta}^{\overset{.}{\alpha}}\left(  \tau\right)  $ depending on
the evolution parameter $\tau.$ The Hamiltonian in square root form,
namely$\sqrt{m^{2}-\mathcal{P}_{0}\mathcal{P}^{0}-\left(  \mathcal{P}%
_{i}\mathcal{P}^{i}+\frac{1}{a}\Pi^{\alpha}\Pi_{\alpha}-\frac{1}{a^{\ast}}%
\Pi^{\overset{.}{\alpha}}\Pi_{\overset{.}{\alpha}}\right)  }\left\vert
\Psi\right\rangle =0$, was constructed defining the supermomenta as usual and,
due the nullification of this Hamiltonian, the Lanczos method for constrained
Hamiltonian systems was used.

Consequently, we have shown that there exist an algebraic interpretation of
the pseudodifferential operator (square root) in the case of an underlying
Mp$\left(  n\right)  $ group structure%

\begin{equation}
\sqrt{m^{2}-\mathcal{P}_{0}\mathcal{P}^{0}-\left(  \mathcal{P}_{i}%
\mathcal{P}^{i}+\frac{1}{a}\Pi^{\alpha}\Pi_{\alpha}-\frac{1}{a^{\ast}}%
\Pi^{\overset{.}{\alpha}}\Pi_{\overset{.}{\alpha}}\right)  }\left\vert
\Psi\right\rangle =0 \tag{61}%
\end{equation}%
\begin{equation}
\left\{  \left[  m^{2}-\mathcal{P}_{0}\mathcal{P}^{0}-\left(  \mathcal{P}%
_{i}\mathcal{P}^{i}+\frac{1}{a}\Pi^{\alpha}\Pi_{\alpha}-\frac{1}{a^{\ast}}%
\Pi^{\overset{.}{\alpha}}\Pi_{\overset{.}{\alpha}}\right)  \right]  _{\beta
}^{\alpha}\left(  \Psi L_{\alpha}\right)  \right\}  \Psi^{\beta}=0 \tag{62}%
\end{equation}

\bigskip then, both structures can be identified: e.g. $\sqrt{m^{2}%
-\mathcal{P}_{0}\mathcal{P}^{0}-\left(  \mathcal{P}_{i}\mathcal{P}^{i}%
+\frac{1}{a}\Pi^{\alpha}\Pi_{\alpha}-\frac{1}{a^{\ast}}\Pi^{\overset{.}%
{\alpha}}\Pi_{\overset{.}{\alpha}}\right)  }\leftrightarrow\left[
m^{2}-\mathcal{P}_{0}\mathcal{P}^{0}-\left(  \mathcal{P}_{i}\mathcal{P}%
^{i}+\frac{1}{a}\Pi^{\alpha}\Pi_{\alpha}-\frac{1}{a^{\ast}}\Pi^{\overset
{.}{\alpha}}\Pi_{\overset{.}{\alpha}}\right)  \right]  _{\beta}^{\alpha
}\left(  \Psi L_{\alpha}\right)  $ being the state $\Psi$ the square root of a
spinor $\Phi$(where the "square root" Hamiltonian acts) such that it can be
bilinearly defined as $\Phi=\Psi L_{\alpha}\Psi.$ Our goal in these references
was based on the observation that the operability of the pseudodifferential
"square root" Hamiltonian can be clearly interpreted if it acts on the square
root of the physical states. In the case of the Metaplectic group, the square
root of a spinor certainly exist [\cite{meta}, \cite{dir}\cite{sann}%
\cite{Major}] making this interpretation (61-62) fully consistent from the
relativistic and \ group theoretical viewpoint.

It is interesting to note, that in ref.\cite{Dattoli} the Dirac factorization
of the one dimensional relativistic Schrodinger equation was treated
introducing the so called quantum simulation of the Dirac equation \cite{Nat}.
This is, in effect, a toy model apparently capable to simulate a genuine
quantum relativistic effect, as the Zitterbewegung. However, the vector
$\underline{\Psi}$ in \cite{Dattoli} is not an spinor and eq. (37) from
\cite{Dattoli} is not the relativistic counterpart of the Pauli equation:
there are not spin degrees of freedom and relativistic invariance. \ The
construction given there is only a mathematical artifact in order to mimify
the relativistic effects in a sharp contrast with equation (62) that is fully
relativistic and capable of include a complete (super) multiplet (spanning
spins from 0 , 1/2, 1, 3/2, 2 ) of physical states. In the next paragraph, we
will describe these states (truly spinorial and relativistic ones) coming from
the algebraic correspondence in order to compare they with the respective
results of the quantum simulation of the Dirac equation results presented in
\cite{Dattoli}.

\subsection{Superspinorial Zitterbewegung}

Now we will pass to analyze and review the description given in \cite{epjc} to
see the origin of the quantum relativistic effects as the Zitterbewegung.
Concerning to the solutions obtained in the \ "algebro-pseudodifferential"
correspondence, we must regard that there are two types of states: the basic
(non-observable) ones and observable physical states (see from another point
of view the results of refs.\cite{diego1}\cite{diego2}\cite{diego3}%
\cite{epjc}). The basic states are coherent states corresponding to the double
covering of the $SL(2C)$ or the metaplectic group\cite{diego1}\cite{diego2}%
\cite{diego3}\cite{epjc} responsible for projecting the symmetries of the 6
dimensional $Mp(4)$ group space to the 4 dimensional spacetime by means of a
bilinear combination of the $Mp(4)$ generators.

Regarding previous works \cite{diego1, diego2}, the supermultiplet solution
for the geometric lagrangian was%
\[
g_{ab}(0,\lambda)=\left\langle \psi_{\lambda}\left(  t\right)  \right\vert
L_{ab}\left\vert \psi_{\lambda}\left(  t\right)  \right\rangle =e^{-\left(
\frac{m}{\left\vert a\right\vert }\right)  ^{2}t^{2}+c_{1}t+c_{2}}%
e^{\xi\varrho\left(  t\right)  }\chi_{f}\langle\psi_{\lambda}(0)|\left(
\begin{array}
[c]{c}%
a\\
a^{\dagger}%
\end{array}
\right)  _{ab}|\psi_{\lambda}(0)\rangle
\]
where we have been written the corresponding indices for the simplest
supermetric state solution being $L_{ab}$ the corresponding generators $\in
Mp\left(  n\right)  $ in the representation given in\cite{sann} \cite{meta}and
$\chi_{f}$, coming from the odd generators of the big covering group related
to the symmetries of the specific model \ (will not be treated in this paper,
and will be left aside). Consider, for simplicity, the `square' solution for
the three compactified dimensions \cite{diego2} (spin $\lambda$ fixed,
$\xi\equiv-\left(  \overline{\xi}^{\overset{.}{\alpha}}-\xi^{\alpha}\right)
$) We have obtained schematically for the exponential even fermionic part
\begin{align}
\varrho\left(  t\right)  \equiv\overset{\circ}{\phi}_{\alpha}\left[  \left(
\alpha e^{i\omega t/2}\right.  \right.   &  \left.  \left.  +\right.  \right.
\left.  \left.  \beta e^{-i\omega t/2}\right)  -\left(  \sigma^{0}\right)
_{\overset{.}{\alpha}}^{\alpha}\left(  \alpha e^{i\omega t/2}-\beta
e^{-i\omega t/2}\right)  \right] \tag{63}\\
&  +\frac{2i}{\omega}\left[  \left(  \sigma^{0}\right)  _{\alpha}^{\overset
{.}{\ \beta}}\ \overline{Z}_{\overset{.}{\beta}}+\left(  \sigma^{0}\right)
_{\ \overset{.}{\alpha}}^{\alpha}\ Z_{\alpha}\right]  \tag{64}%
\end{align}

where $\overset{\circ}{\phi}_{\alpha},Z_{\alpha},\overline{Z}_{\overset
{.}{\beta}}$ are constant spinors, and $\alpha$ and $\beta$ are $\mathbb{C}%
$-numbers (the constant $c_{1}\in\mathbb{C}$ due the obvious physical reasons
and the chirality restoration of the superfield solution [1,2,10]).

By consistency, as in the case of the string, two geometric-physical options
will be related to the orientability of the superspace trajectory\cite{ger}:
$\alpha=\pm\beta$. We have take, without lose generality $\alpha=+\beta$ then,
exactly, there are two possibilities:

i) the compact case (which was given before in \cite{diego2}\cite{diego1})%
\begin{equation}
\varrho\left(  t\right)  =\left(
\begin{array}
[c]{c}%
\overset{\circ}{\phi}_{\alpha}\cos\left(  \omega t/2\right)  +\frac{2}{\omega
}Z_{\alpha}\\
-\overset{\circ}{\overline{\phi}}_{\overset{\cdot}{\alpha}}\sin\left(  \omega
t/2\right)  -\frac{2}{\omega}\overline{Z}_{\overset{.}{\alpha}}%
\end{array}
\right)  \tag{65}%
\end{equation}
ii) and the non-compact case%
\begin{equation}
\varrho\left(  t\right)  =\left(
\begin{array}
[c]{c}%
\overset{\circ}{\phi}\cosh\left(  \omega t/2\right)  +\frac{2}{\omega
}Z_{\alpha}\\
-\overset{\circ}{\overline{\phi}}_{\overset{\cdot}{\alpha}}\sinh\left(  \omega
t/2\right)  -\frac{2}{\omega}\overline{Z}_{\overset{.}{\alpha}}%
\end{array}
\right)  \tag{66}%
\end{equation}
obviously (in both cases) represents a \textit{Majorana fermion} where the
$\mathbb{C}$ (or $hypercomplex$ wherever the case$)$ symmetry is inside of the
constant spinors.

The spinorial even part of the superfield solution in the exponent becomes to
\begin{equation}
\xi\varrho\left(  t\right)  =\theta^{\alpha}\left(  \overset{\circ}{\phi
}_{\alpha}\cos\left(  \omega t/2\right)  +\frac{2}{\omega}Z_{\alpha}\right)
-\overline{\theta}^{\overset{\cdot}{\alpha}}\left(  -\overset{\circ}%
{\overline{\phi}}_{\overset{\cdot}{\alpha}}\sin\left(  \omega t/2\right)
-\frac{2}{\omega}\overline{Z}_{\overset{.}{\alpha}}\right)  \tag{67}%
\end{equation}
for the $\mathbb{C}$ (or $hypercomplex$ wherever the case$)$ symmetry. We
easily see that in the above expression there appear a type of
\textit{Zitterbewegung} or continuous oscillation between the chiral and
antichiral part of the bispinor $\varrho(t)$. ( see for example, Figures 1, 2,
3 and 4 in ref. \cite{epjc} are snapshots describing the time evolution of the
oscillating effect for suitable values of the parameters of the vacuum
solution and with an increasing $\omega t\sim t/\left\vert a\right\vert $
respectively ($\omega_{1}<\omega_{2}<\omega_{3}....$)).

\begin{remark}
the physical meaning of such an oscillation (Zitterbewegung) is simply an
underlying natural supersymmetric effect because there exists a kind of
duality between supersymmetrical and relativistic effects, pointed out
previously in \cite{casa}.
\end{remark}

\section{Physical examples}

\subsection{Rogue waves}

As is more or less understandable from their scientific observation at the
Draupner oil platform in the North sea\cite{Dra}, rogue waves (sometimes
described as monster waves, freak waves or giant waves) appear with an
amplitude extremely larger with respect to the amplitude of the surrounding
wave crests \cite{Osb}. Because the conditions that cause the enormous growth
of rogue waves is still not well known, they have become a subject of intense
research after their experimental realization and simulation in various
physical systems having an underlying nonlinear character, namely, optical
fibers \cite{fo1}\cite{fo2}, plasmas\cite{pl}.

The general form to attack the problem in many of these contexts is to
introduce different variants of the nonlinear Schrodinger equation (NLSE)\ due
the modulation of its present instability \cite{in}. This kind of modulation
can be effectively implemented at laboratory level in the Bose-Einstein
condensation (BEC) scenarios where the Feshbach resonance technique
\cite{Fesh} allows to control the dynamics of matter rogue waves by mean the
feasibility of tuning interatomic interactions. Also, in the same context of
BEC, the quasi-one dimensional Gross-Pitaevskii (GP)\ as the NLSE\ with
trapping potential is usually utilized.

The interesting point is that in recent years the study of spinor condensates
has been an important issue experimentally and theoretically speaking
\cite{spinorcond} As is well known , he dynamics of the spinor condensate is
described within the mean-field approximation by the multicomponent
GP-equations containing nontrivial nonlinear terms mimifiying the SU(2)
symmetry of the spins. Alkali-metal atoms are usually represented of a such
manner. Now, we can see that rogue waves can be obtained by the free
Schrodinger equation solution in the algebro-pseudodifferential framework
given above. Controlling the parameters of the solution given by expressions
(63-67), the expected waves are obtained when the time variable have complex
coefficients (e.g$\rightarrow it)$. It is clear that it is a kind of "
dispersion in time" that is responsible of the "exploded" wave dynamics .
3-Dimensional figures 5,6 and 7 show the wave behaviour with an increasing
imaginary coefficient in the time variable, respectively. Figures 1,2,3 and 4
show snapshots describing the evolution of the oscillating effect in complex time..

\subsection{The Nambu-Goto action and the microcanonical propagator}

Here we will make some\ comments about the pseudodifferential operators and
physical systems with finite energy (e.g. microcanonical ensemble) in
connection with the quantum field theoretical (QFT) viewpoint. Regarding our
previous reference (see full details in \cite{micro} and connected with a
string-high energy framework see \cite{norma} )the propagator for
black-hole/string/particle was constructed knowing that the Nambu-Goto
action\ is invariant under the reparametrizations. Using the "Born-Infeld"
choice for the dynamical variables \cite{barb} we obtain the action in the
form
\begin{equation}
S=-\frac{\kappa}{\alpha^{\prime}}\int_{\tau1}^{\tau2}\overset{.}{x}_{0}%
d\sigma\ d\tau\ \sqrt{\left[  1-\left(  \partial_{0}x_{b}\right)  ^{2}\right]
\left[  1+\left(  \partial_{1}x_{a}\right)  ^{2}\right]  },\text{
\ \ \ \ \ \ \ \ \ \ \ \ \ \ }\left(  a,b=2,3;\ \partial_{1}x_{a}%
=\varepsilon_{1a}^{\ \ \ 0b}\partial_{0}x_{b}\right)  , \tag{68}%
\end{equation}

Therefore, the invariance with respect to the invariance of the coordinate
evolution parameter means that one of the dynamic variables of the theory
($x_{0}\left(  \tau\right)  $ in this case) becomes the observed time with the
corresponding non-zero Hamiltonian $H_{BI}=\Pi_{a}\overset{.}{x}^{a}%
-L=\sqrt{\alpha^{2}-\Pi_{b}\Pi^{b}},$where: $\Pi^{b}=\frac{\partial
L}{\partial\left(  \partial_{0}x_{b}\right)  },\alpha\equiv\frac{\kappa
\sqrt{1+\left(  \partial_{1}x_{a}\right)  ^{2}}}{\alpha^{\prime}}$ From the
simplest path-integral formalism, using quantum field theoretical arguments
and introducing the integral representation for a pseudodifferential operator
\cite{brich}(based in semigroup construction)
\begin{equation}
\int\left(  t^{2}+u^{2}\right)  ^{-\lambda}e^{itx}dt=\frac{2\pi^{1/2}}%
{\Gamma\left(  \lambda\right)  }\left(  \frac{\left\vert x\right\vert }%
{2u}\right)  ^{\lambda-1/2}K_{\lambda-1/2}\left(  u\left\vert x\right\vert
\right)  \tag{69}%
\end{equation}
where $K_{\nu}\left(  x\right)  $ is the MacDonald's function, we obtain
\cite{norma}\cite{micro}the following microcanonical propagator:
\[
D_{E}(t,\overline{x})=\frac{\delta\left(  E\right)  }{\omega^{2}-k^{2}%
-m^{2}+i\varepsilon}-
\]%
\begin{equation}
-8\pi i\alpha\delta\left(  \omega^{2}-k^{2}-m^{2}\right)  \overset{}%
{\underset{l=1}{\overset{\infty}{\sum}}}\frac{K_{-1}\left(  \alpha\left\vert
l\omega_{k}-E\right\vert \right)  }{E^{2}}\frac{\Omega\left(  E-l\omega
_{k}\right)  }{\Omega\left(  E\right)  }\theta\left(  E-l\omega_{k}\right)  ,
\tag{70}%
\end{equation}
where $\theta\left(  x\right)  $ is the usual step function. The first term in
the microcanonical propagator is the usual (non-termal) Feynman propagator,
the second one is the new microcanonical statistical part. The correct
description of the full N-extended body system is obtained explicitly
expanding the Mac Donald's function $K_{-1}$ in the second term of the free
microcanonical propagator \cite{norma}leading a nonlocal and nonlinear
generalization of the well known \ (string-theoretical) Veneziano amplitude
\cite{ven}. This observation leads us to highlight the next:

\begin{remark}
pseudo differential operators in QFT give rise to propagators which a
string-like type of structure emerges (Gamma type
string-amplitude)contributing to their statistical the relation between
temporal and normal ordering of the field operators.
\end{remark}

\subsection{Warped gravities, Randall Sundrum scenarios and the square root}

The last example coming from our reference\cite{diego2} showing itself the
consistency of this interpretation. The motivation to introduce
pseudodifferential operators was to find the consistent solution to the
hierarchy problem\cite{hier} and, due the lack of formal "first principles"
explanations, to the field theoretical localization mechanisms for scalar and
fermions \cite{shap} as well as for gauge bosons \cite{dv}. Some points coming
from the analisys of these previous works \cite{diego2}must be highlighted:

The remarkable property of the\ full solution involving beside the expressions
(63-67), the bosonic part namely: $g_{ab}(t)=e^{-\left(  \frac{m}%
{|{\mathbf{a}}|}\right)  ^{2}t^{2}+c_{1}t+c_{2}+\xi\varrho\left(  t\right)
}g_{ab}(0)$ \cite{diego2}is that the physical state $g_{ab}\left(  x\right)  $
is localized in a particular position of the space-time: the supermetric
$\mathbb{C}$ coefficients $\mathbf{a}$ $\left(  \mathbf{a}^{\ast}\right)
$play the important role of localize the fields in the bosonic part of the
superspace in similar and suggestive form as the well known "warp factors" in
multidimensional gravity\cite{bajc} for a positive (or negative) tension
brane. This Gaussian type solution is very well defined physical state in a
Hilbert space\cite{diego4}\cite{kl}from the mathematical point of view,
contrarily to the sual case $u\left(  y\right)  =ce^{-H\left\vert y\right\vert
}$ given in\cite{bajc} that, although were possible to find a manner to
include it in any Hilbert space, is strongly needed to take special
mathematical and physical particular assumptions whose meaning is obscure. For
a more complete picture the comparison with the case of 5-dimensional gravity
plus cosmological constant\cite{bajc} is clearly given with full details in
the table of reference\cite{diego2}\footnotetext[2]{the extended superspace
solution in the case contain all the 4-dimensional coordinates: $x\equiv
\left(  t,\overline{x}\right)  $, $c_{1}^{\prime}x\equiv c_{1\mu}^{\prime
}x^{\mu}$ and $c_{2}^{\prime}$ scalar.} .

\section{Concluding remarks}

In this work, we have been made a development and analysis of the problem
generated by the "square root" operator. Through this paper we logically
emphasize the non-locality and the relation with the Green function approach.

We show that there exists a close relation between the number of the spacetime
dimensions and the order of the of the cylinder functions ( MacDonald's
function in our case) having the case of even number of spacetime dimensions
an exact integrability.

The self-reproducing property of the Mc Donald's function, as the main
ingredient of the Green kernel, makes to be possible the straighforward
relation with the coherent and squeezed states of the non-compact groups
(e.g.$SU(1,1)$) and the corresponding double coverings as is the Metaplectic
the typical case \cite{meta}. In this sense, we have demonstrate here clearly,
through the comparison with references \cite{various}, the relation with the
non-hermitian time operator and, looking at reference\cite{yosi} , the
relation with the time-energy coherent states. In this manner we have been
shown specifically the form of the "overlap" integrals and the physical
operators of the "observables": velocity and the phase space structure. The
conclusion that is immediately obtained from this last point is that the
coherent states structure is related with the time dispersion of the non-local
square root Hamiltonian.

The integrability for several dimensions is achieved in the case of even
dimensions due the symplectic (phase space) underlying structure. All
remaining aspects concerning these issues were clearly treated through the
first part of the paper by mean the respective theorems, proofs and remarks.

The algebraic connection with the pseudodifferential description, described
from our previous works and formally proposed by us here, allows the correct
interpretation of the square root treated as operator that have been
exemplified by three physical cases (namely: rogue waves, warped gravities and
the Nambu-Goto action and the microcanonical propagator). The relativistic
wave equations for any spin is also described by this algebraic interpretation
establishing a bridge with the pseudodifferential and semigroup approaches.

An important new result, in the context of the algebraic approach, that we
have found before is that there exist an oscillatory fermionic effect in the
$B_{0}$ part of the supermultiplet as a Zitterwebegung, but between the chiral
and antichiral components of this Majorana bispinor. This effect is (see
equation (66) )is fully relativistic and capable of include a complete (super)
multiplet of physical states in a sharp contrast with ref.\cite{Dattoli} where
the Dirac factorization of the one dimensional relativistic Schrodinger
equation was treated introducing the so called quantum simulation of the Dirac
equation \cite{Nat}. where the vector $\underline{\Psi}$ in \cite{Dattoli} is
not an spinor and eq. (37) from \cite{Dattoli} is not the relativistic
counterpart of the Pauli equation: there are not spin degrees of freedom and
relativistic invariance.

\section{Acknowledgements}

Many thanks are given to Professors Yu. P. Stepanovsky, John Klauder and E. C.
G. Sudarshan for their interest; and particularly to Professor N. Mukunda for
several discussions in the metaplectic symmetries. I am very grateful to the
CONICET-Argentina and also to the BLTP-JINR Directorate for their hospitality
and financial support for part of this work.

\section{}

\bigskip

\bigskip

\end{document}